\documentclass[11pt,thmsa,twoside]{article}

\usepackage[english]{babel}					
\usepackage[T1]{fontenc}						
\usepackage[utf8]{inputenc}					
\usepackage{lmodern}							

\usepackage{amssymb,amsmath,amsfonts,mathtools,amsthm}	
\usepackage{latexsym}						
\usepackage{geometry}
\usepackage{graphics}						
\usepackage{graphicx}						
\usepackage{wrapfig}							
\usepackage{color}							
\usepackage{savefnmark}						

\usepackage{fancyhdr}						
\usepackage{setspace}					 	
\usepackage{url}										
\usepackage{fancyhdr}
\usepackage{xfrac}	
\usepackage{enumitem} 
\usepackage{mathrsfs} 
\usepackage{lmodern}			
\usepackage{booktabs} 
\usepackage{tikz} 
\usetikzlibrary{matrix}
\usetikzlibrary{arrows} 
\usepackage{proof} 
\usepackage{hyperref}
\usepackage{bussproofs}
\usepackage{bookmark}
\usepackage{verbatim}
\usepackage{array} 
\usepackage{titlesec} 
\usepackage[lf]{ebgaramond}

\newcommand{\fakebf}{\fontfamily{ptm}\fontseries{b}\selectfont}
\DeclareTextFontCommand{\textbf}{\fakebf}
							
\newtheorem{theor}{\textbf{Theorem}}[section] 
\newtheorem{prop}{\textbf{Proposition}}[section] 

\newcommand{\ttt}{\mbox{\bf t}}
\newcommand{\ff}{\mbox{\bf f}}

\newcommand{\acc}{\ensuremath{\mathsf{Y}}}
\newcommand{\rej}{\ensuremath{\mathsf{N}}}
\newcommand{\nacc}{{\mathrel{\rotatebox[origin=c]{180}{$\acc$}}}}
\newcommand{\nrej}{{\mathrel{\reflectbox{$\rej$}}}}

\newcommand{\ag}[1]{s_{#1}}
\newcommand{\agn}{\ag{}}

\newcommand{\Efrac}[2]{%
  \mathchoice
    {\ooalign{%
      $\genfrac{}{}{1.2pt}0{#1}{#2}$\cr%
      $\color{white}\genfrac{}{}{.4pt}0{\phantom{#1}}{\phantom{#2}}$}}%
    {\ooalign{%
      $\genfrac{}{}{1.2pt}1{#1}{#2}$\cr%
      $\color{white}\genfrac{}{}{.4pt}1{\phantom{#1}}{\phantom{#2}}$}}%
    {\ooalign{%
      $\genfrac{}{}{1.2pt}2{#1}{#2}$\cr%
      $\color{white}\genfrac{}{}{.4pt}2{\phantom{#1}}{\phantom{#2}}$}}%
    {\ooalign{%
      $\genfrac{}{}{1.2pt}3{#1}{#2}$\cr%
      $\color{white}\genfrac{}{}{.4pt}3{\phantom{#1}}{\phantom{#2}}$}}%
}

\newcommand{\Bcon}[4]{\frac{#2}{#1}{\big|}\frac{#4}{#3}}
\newcommand{\Bent}[4]{\Efrac{#2}{#1}{\big|}\Efrac{#4}{#3}}

\newcommand{\nBent}[4]{\Efrac{#2}{#1}{\mathclap{\smash{\,\,\times}}\big|}\Efrac{#4}{#3}}

\newcommand{\BconCONT}[4]{\Bcon{#2\Gamma}{#1\Psi}{\Phi#4}{\Delta#3}}

\titleformat{\section}
{\fontfamily{ptm}\bfseries}
{\thesection.}{0.2em}{}
 
\titlespacing{\section}{0pt}{5ex plus .1ex minus .2ex}{1pc}

\renewenvironment{quote}%
  {\list{}{\leftmargin=2cm\rightmargin=0cm}\item[]}%
  {\endlist}

\title{Revisiting the Dunn-Belnap logic}
\author{Carolina Blasio}
\date{translated by Evelyn Erickson\\
revised by Jo\~ao Marcos}

\usepackage{setspace} 

\begin{document}

\maketitle


\begin{abstract}
\noindent In the present work I introduce a semantics based on the cognitive attitudes of acception and rejection entertained by a given society of agents for logics inspired on Dunn and Belnap’s First Degree Entailment ($\mathbf{E}$). In contrast to the epistemic situations originally employed by $\mathbf{E}$, the cognitive attitudes do not coincide with truth-values and they seem more suitable to logics that intend to consider the informational content of propositions ``said to be true'' as well as of propositions ``said to be false'' as determinant of the notion of logical validity. After analyzing some logics associated to the proposed semantics, we introduce the logic $\mathbf{E}^B$, whose underlying entailment relation---the $B$-entailment---is able to express several kinds of reasoning involving the cognitive attitudes of acceptance and rejection. A sound and complete sequent calculus for $\mathbf{E}^B$ is also presented.\footnote{~Originally published as: Blasio, Carolina. Revisitando a Lógica de Dunn-Belnap. \emph{Manuscrito}, v.40, n.2, 2018, \href{http://dx.doi.org/10.1590/0100-6045.2017.V40N2.CB}{DOI 10.1590/0100-6045.2017.V40N2.CB}.
}

\end{abstract}

\section{The Dunn-Belnap Logic}

In the mid 1970s, when research on artificial intelligence was consolidated, the demand emerged that a computer, in addition to being able to answer questions using deductive reasoning, should be able to deal with data even if they were to contain some inconsistency (explicit or otherwise) or partiality.
The creation of such a hypothetical artificial reasoner inspired the logic $\mathbf{E}$ of {First Degree Entailment}, or Dunn-Belnap logic \cite{Dunn1976,Belnap1977}, a relevant logic in which the classical principles of excluded middle and of explosion are not valid. 

The logic $\mathbf{E}$, as we will see in the following, has a semantics with four non-classical truth-values which allow, in a certain way, to deal with inconsistencies and with partial informational content. The logic $\mathbf{E}$, however, has a consequence relation defined in terms of the preservation of the set of truth-values which are assigned to the statements said to be true. This limitation, inherent to Tarskian consequence relations, does not seem adequate to a form of reasoning which intends to deal with statements said to be both true as well as with those said to be false, given that there is no complementarity 
between that which is said to be true and that which is said to be false.

The Dunn-Belnap logic is given by the structure:
 \[\mathbf{E}=\langle\mathnormal{S},\vDash^4\rangle\] 
 where $\mathnormal{S}$ is a language recursively formed by symbols of propositional variables, by the unary connective \emph{negation} ($\neg$) and by the binary logical connectives \emph{disjunction} ($\vee$) and \emph{conjunction} ($\wedge$), and where $\vDash^4$ is an entailment (here taken to be synonymous to `semantic consequence') relation.

The logic $\mathbf{E}$ has a semantics with four values structured in a bilattice known as $\mathit{FOUR}$. A lattice is a structure  $\langle B, \leq \rangle$ such that the elements of the set $B$ are ordered by the partial order relation 
$\leq$ and each pair $x,~y\in B$ has a supremum ($\sqcup$) and an infimum ($\sqcap$), defined by 
\[x\sqcap y=x\textrm{ ~iff~ }x\leq y,\]
\[ x\sqcup y=x\textrm{ ~iff~ }y\leq x.\]
A \emph{bilattice} is an algebraic structure $\mathfrak{B}=\langle B, \leq_t, \leq_i,-_t \rangle$ such that $\langle B, \leq_t \rangle$ and $\langle B, \leq_i \rangle$ are both lattices, and the inverse $-_t$ is a unary operation which satisfies the clauses that follow: for all $x,~y\in B$,
\begin{center}
$\mathrm{if }~ x\leq_t y $ ,~then~ $-_ty\leq_t-_tx,$
\end{center}
\begin{center}
$\mathrm{if }~ x\leq_i y $ ,~then~ $ -_tx\leq_i-_ty,$
\end{center}
\begin{center}
$x=-_t-_tx.$
\end{center}
The semantics of the logic $\mathbf{E}$ is given by the bilattice
\[\mathit{FOUR}{=}\langle~\mathbf{4},~\leq_t,~\leq_i,~-_t~\rangle,\] where $\mathbf{4}=\{\ff,~\bot,~\top,~\ttt\}$. We represent $FOUR$ by the Hasse diagram in the Figure \ref{fig:1} below.
%
\begin{figure}[ht]
\begin{center}
\ifx\du\undefined
  \newlength{\du}
\fi
\setlength{\du}{15\unitlength}
\begin{tikzpicture}
\pgftransformxscale{.85000000}
\pgftransformyscale{-.85000000}
\draw (25.000000\du,5.000000\du)--(27.000000\du,7.000000\du)--(25.000000\du,9.000000\du)--(23.000000\du,7.000000\du)--cycle;
\draw (25.000000\du,5.000000\du)--(27.000000\du,7.000000\du)--(25.000000\du,9.000000\du)--(23.000000\du,7.000000\du)--cycle;
\node[anchor=west] at (27.000000\du,7.000000\du){$\ttt$};
\node[anchor=west] at (22.00000\du,7.000000\du){$\ff$};
\node[anchor=west] at (24.3500000\du,4.500000\du){$\top$};
\node[anchor=west] at (24.3500000\du,9.500000\du){$\bot$};
{
\pgfsetarrowsstart{to}
\draw (22.000000\du,4.000000\du)--(22.000000\du,10.000000\du);
}
{
\pgfsetarrowsend{to}
\draw (22.000000\du,10.000000\du)--(28.000000\du,10.000000\du);
}
\node[anchor=west] at (20.50000\du,4.500000\du){$\leq_i$};
\node[anchor=west] at (27.00000\du,10.60000\du){$\leq_t$};
\end{tikzpicture}
\end{center}
\vspace{-1cm}
	\caption{$\mathit{FOUR}$}
	\label{fig:1}
\end{figure}

Each element of $\mathbf{4}$ corresponds to an element of the powerset $\wp(\{T, F\})$ of the set of the classical truth-values $T$ and $F$. These four truth-values have been referred, by Belnap, as ``epistemic situations''  
in which the informational content of a given proposition $\varphi$ is represented by the value given in Table \ref{tab:se}.
\begin{table}[ht]
\begin{center}
\small
\begin{tabular}{l@{ := }ll}\vspace{.2cm}
$\ff$&$\{F\}$ & if $\varphi$ \emph{is said to be only false};\\\vspace{.2cm}
$\bot$&$\emptyset$ & if $\varphi$ \emph{is neither said to be true nor said to be false};\\\vspace{.2cm}
$\top$&$\{F, T\}$ & if $\varphi$ \emph{is said to be true and said to be false} ;\\\vspace{.2cm}
$\ttt$&$\{T\}$& if $\varphi$ \emph{is said to be only true}.
\end{tabular}
\caption{\emph{Truth-values in terms of epistemic situations.}}
\label{tab:se}
\end{center}
\vspace{-.6cm}
\end{table}

\noindent Notice that the reading of the truth-values given in Table \ref{tab:se} implies that there are situations in which a given statement is considered ``neither said to be true nor said to be false'', or ``both said to be true and said to be false''. We should call attention as well to the fact that the values $\ff$ and $\ttt$ are not the classical values $F$ and $T$.

The order $\leq_t$ of $\mathit{FOUR}$ is, thus, known as the ``logical order'' because the elements of $\mathbf{4}$ are ordered from the ``said to be more false'' to the ``said to be more true''. The ``information order'' $\leq_i$, in turn, is generally understood as the order which goes from the ``lack of information'' to the ``excess of information''.\footnote{While this order often appears in the literature with the name of `knowledge order'
, we agree with \cite[p. 3]{Fitting2006} that the term `information' is more appropriate since it is more neutral with respect to the concepts of belief and truth.}

Formally,
 \[ x\leq_ty \textrm{ iff } x^{\ttt}\subseteq y^{\ttt}\textrm{ and }y^{\ff}\subseteq x^{\ff},\] for each $x, y\in\mathbf{4}$, given $a^{\ttt}=\{T\}\cap a$ (the true part of $a$) and $a^{\ff}=\{F\}\cap a$ (the false part of $a$); and
  \[x\leq_iy \textrm{ iff } x\subseteq y.\]

A valuation $v^4{:}\,\mathnormal{S}\to\mathbf{4}$ based in $FOUR$ ~is a homomorphism from the language $\mathnormal{S}$ to the structured truth-values in~$\mathbf{4}$. The semantics $SEM^{\mathbf{E}}$ of $\mathbf{E}$ is formed by all the valuations based on~$FOUR$. We will abbreviate by $v^4(\Gamma)=\{v^4(\gamma):\gamma\in\Gamma\}$ the result of applying of a valuation $v^4$ on a set of formulas $\Gamma\subseteq\mathcal{S}$.

      The symbol $\vDash^4$ represents the entailment relation 
      of $\mathbf{E}$, defined as a subset of $\wp(\mathnormal{S})\times\mathnormal{S}$ such that, for all $\alpha\in\mathnormal{S}$ and all $\Gamma\subseteq\mathnormal{S}$, 
  \[\Gamma\vDash^4\alpha\textrm{ ~iff~for all }v^4\in SEM^{\mathbf{E}},~~\sqcap_tv^4 (\Gamma)\leq_t v^4(\alpha), \]
  where $\sqcap_tv^4 (\Gamma)$ is the infimum of the set of truth-values assigned to the formulas of $\Gamma$ with respect to the order $\leq_t$. We can understand the definition of validity of $\vDash^4$ as: the inference $\Gamma\vDash^4\alpha$ is valid if, and only if, $\alpha$ is said to be at least as true as all the statements of $\Gamma$.
  
 An equivalent way \footnote{The equivalence between these definitions of consequence for $\mathbf{E}$ is proved in \cite{Arieli1998}, Proposition 4.14.} of defining $\vDash^4$ is from the logical matrix $\mathfrak{M}=\langle \mathbf{4}, \mathcal{F}, \mathcal{O}\rangle$, where $\mathbf{4}$ is the previously defined set of values, $\mathcal{F}$ is the set of designated values corresponding to the (prime) bifilter\footnote{The bifilter of $\mathit{FOUR}$ is the non-empty subset $\mathfrak{F}\subset\mathbf{4}$ such that, for each $x, y\in \mathbf{4}$,
  
   $x\sqcap_ty\in\mathfrak{F}$ iff $x\in\mathfrak{F}$ and $y\in\mathfrak{F}$, and
   
    $x\sqcap_iy\in\mathfrak{F}$ iff $x\in\mathfrak{F}$ and $y\in\mathfrak{F}$.\\ The bifilter $\mathfrak{F}$ of $\mathit{FOUR}$ is called \emph{prime} if 
    
    $x\sqcup_ty\in\mathfrak{F}$ iff $x\in\mathfrak{F}$ or $y\in\mathfrak{F}$, and 
    
    $x\sqcup_iy\in\mathfrak{F}$ iff $x\in\mathfrak{F}$ or $y\in\mathfrak{F}$. }
of $\mathit{FOUR}$, that is, $\mathcal{F}=\{\top, \ttt\}$, and the set $\mathcal{O}=\{\sqcap_t, \sqcup_t, -_t\}$ contains the operations of inversion, infimum and supremum of the order $\leq_t$ of $FOUR$, which correspond to the  truth-functions of each of the connectives of the language $\mathnormal{S}$, defined below:
  \[v^4(\neg\alpha)=-_tv^4(\alpha)\]
  \[v^4(\alpha\wedge\beta)=v^4(\alpha)\sqcap_t v^4(\beta)\]
  \[v^4(\alpha\vee\beta)=v^4(\alpha)\sqcup_t v^4(\beta)\]
  
 The consequence relation $\vDash^4$ associated to $\mathfrak{M}$ is such that, for all $\alpha\in\mathnormal{S}$ and for all $\Gamma\subseteq\mathnormal{S}$,
\[\Gamma\vDash^4\alpha\textrm{ iff there is no }v^4\in SEM^{\mathbf{E}} \textrm{ such that }v^4(\Gamma)\subseteq\mathcal{F}\textrm{ and }v^4(\alpha)\in\mathbf{4}-\mathcal{F}.\]

As with all Tarskian consequence relations, $\vDash^4$ respects the properties of Reflexivity, Monotonicity and Transitivity: where $\alpha,~\beta\in\mathnormal{S}$ and $\Delta,~\Gamma\subseteq\mathnormal{S}$,

\begin{description}

\item[\textbf{Reflexivity}]\textcolor{white}{}
$\Delta\cup\{\alpha\}\vDash^4\alpha$

\item[\textbf{Monotonicity}]\textcolor{white}{}
$\textrm{If }\Delta\vDash^4\alpha\textrm{~then~}\Delta\cup\Gamma\vDash^4\alpha$

\item[\textbf{Transitivity}]\textcolor{white}{}
$\textrm{If }\Delta\vDash^4\alpha\textrm{~and, for all~}\beta\in\Delta,~ \Gamma\vDash^4\beta\textrm{,~then~}\Gamma\vDash^4\alpha$
\end{description}

Notice that in the consequence relation $\vDash^4$ associated to $\mathfrak{M}$ only the values which have $T$, that is, only the values having ``the truth'', are preserved from the premises to the conclusion. This provides evidence to the fact that the logical order $\leq_t$ is the only one determining the consequence relation of $\mathbf{E}$.

I claim, however, that a logic that intends to deal with the informational content of statements, even in the presence of inconsistency or of partiality, should not focus only on the propositions ``said to be true'', in detriment of the propositions ``said to be false'', in determining the notion of validity of inferences. The semantics of $\mathbf{E}$ is defined in such a way that the statements ``said to be false'' are not the complement of the statements ``said to be true'', yet both are equally important for the reasoning proposed by the logic $\mathbf{E}$. Accordingly, the notion of validity of $\mathbf{E}$ should also take into consideration that which is ``said to be false''.

I thus propose an alternative reading based on ${FOUR}$ ~in which the definitions based on truth-values give way to the cognitive attitudes of acceptance and rejection of a given informational content by a society of agents. This amounts to choosing the cognitive attitudes of acceptance and rejection as primitive objects of the semantics, in detriment to the truth-values. In what follows, I present some logics, inspired by the Dunn-Belnap logic, which propose to express different forms of reasoning involving the informational content of propositions.

\section{Cognitive attitudes}

We have seen in the previous section that the semantics of the Dunn-Belnap logic $\mathbf{E}$ identifies each truth-value with an epistemic situation (cf.\ Table \ref{tab:se}). Such an identification generates readings of statements which are neither said to be true nor said to be false, and statements which are, at once, both said to be true and said to be false, even though the underlying notion of validity is defined in terms of the preservation of that which is only said to be true.

As an alternative, we propose a semantics formed by a society~$SOC$ of agents which, instead of assigning truth-values directly to the propositions, entertain cognitive attitudes towards the informational content of the conferred propositions: the cognitive attitudes of acceptance and rejection.

In $\mathbf{E}$, the cognitive attitude of an agent may be to ``accept'' ($\acc$), to ``not-accept'' ($\nacc$), to ``reject'' ($\rej$) or to ``not-reject'' ($\nrej$) the informational content of a certain statement $\varphi$. In the case of the logic $\mathbf{E}$, the attitudes of agents are to accept and to reject, but in other cases, the attitudes could be to vote against, to say that it is good, to like, etc.
  
Let $\mathnormal{S}$ be the language of the logic $\mathbf{E}$ and $COG=\{\acc, \nacc,\rej,\nrej\}$ be the set of cognitive attitudes.  An agent $\agn$ is a homomorphism from the language $\mathnormal{S}$ to $COG$. Although the truth-values are no longer primitive objects of the semantics, they may be defined in terms of the cognitive attitudes. We present below the canonical definition of the truth-values starting from the cognitive attitudes. Given an agent $\agn$, a $C\in COG$ and a $\varphi\in\mathnormal{S}$, we read $C\agn{:}\varphi$ as ``the agent $\agn$ entertains the cognitive attitude $C$ towards $\varphi$''. By consulting an agent $s\in SOC$ on the informational content of $\varphi\in\mathnormal{S}$, the truth-value assigned to $\varphi$ will be obtained as shown in Table~\ref{tab:atcog}.
One may also recover the cognitive attitudes from the structured truth-values of $\mathbf{4}$, as shown by Table \ref{tab:vdv}, where $\agn\in SOC$ and $\varphi\in\mathnormal{S}$.

\begin{table}[ht]
\begin{center}

\begin{tabular}{l@{: }l}\vspace{.2cm}
$\mathbf{f}$ &{if $ \nacc\agn{:}\varphi$ and $ \rej\agn{:}\varphi$ }\\\vspace{.2cm}
$\bot$ &{if $ \nacc\agn{:}\varphi$ and $ \nrej\agn{:}\varphi$}\\\vspace{.2cm}
$\top$ &{if $ \acc\agn{:}\varphi$ and $ \rej\agn{:}\varphi$ }\\\vspace{.2cm}
$\mathbf{t}$& {if $ \acc\agn{:}\varphi$ and $ \nrej\agn{:}\varphi$}
\end{tabular}
\end{center}
\vspace{-.6cm}
\caption{\emph{Truth-values in terms of cognitive attitudes.}}
\label{tab:atcog}
\end{table}

\begin{table}[ht]
\begin{center}
$\begin{array}{l@{~~~{\rm iff}~~~}l@{~~~{\rm iff}~~~}l}\vspace{.2cm}
\acc\agn{:}\varphi &T\in\agn(\varphi)   &\agn(\varphi)\in\{\top, \ttt\}\\\vspace{.2cm}
\nacc\agn{:}\varphi&T\notin\agn(\varphi)&\agn(\varphi)\in\{\ff, \bot\}\\\vspace{.2cm}
\rej\agn{:}\varphi &F\in\agn(\varphi)   &\agn(\varphi)\in\{\ff, \top\}\\\vspace{.2cm}
\nrej\agn{:}\varphi&F\notin\agn(\varphi)&\agn(\varphi)\in\{\bot, \ttt\}
\end{array}$
\end{center}
\vspace{-.5cm}
\caption{\emph{Cognitive attitudes in terms of truth-values.}}
\label{tab:vdv}
\end{table}

A structured truth-value which contains $T$ may be assigned to an informational content accepted by an agent,  a structured truth-value which contains $F$ may be assigned to an informational content rejected by an agent. A structured truth-value which does not contain $T$ may be assigned to an informational content not-accepted by an agent, and a structured truth-value which does not contain $F$ may be assigned to an informational content not-rejected by an agent. Notice that, in terms of truth-values, the semantics whose primitive objects are cognitive attitudes is in principle non-deterministic, that is, more than one truth-value may be assigned to the same statement.

Given the cognitive attitude that a given agent $\agn$ entertains towards a propositional variable of the logic $\mathbf{E}$, the compound statements are interpreted by the following {recursive clauses}, where $\varphi, \psi \in \mathnormal{S}$,
\begin{enumerate}[label=2.\arabic*]

\item $\acc\agn{:}\neg\varphi,\textrm{~if~}\rej\agn{:}\varphi$ 

\item $\rej\agn{:}\neg\varphi,\textrm{~if~}\acc\agn{:}\varphi$ 

\item $\nacc\agn{:}\neg\varphi,\textrm{~if~}\nrej\agn{:}\varphi$ 

\item $\nrej\agn{:}\neg\varphi,\textrm{~if~}\nacc\agn{:}\varphi$ 

\item $\acc\agn{:}\varphi\wedge\psi,\textrm{~if~}\acc\agn{:}\varphi \textrm{~and~}\acc\agn{:}\psi$

\item $\rej\agn{:}\varphi\wedge\psi,\textrm{~if~}\rej\agn{:}\varphi \textrm{~or~}\rej\agn{:}\psi$

\item $\nacc\agn{:}\varphi\wedge\psi,\textrm{~if~}\nacc\agn{:}\varphi \textrm{~or~}\nacc\agn{:}\psi$

\item $\nrej\agn{:}\varphi\wedge\psi,\textrm{~if~}\nrej\agn{:}\varphi \textrm{~and~}\nrej\agn{:}\psi$

\item $\acc\agn{:}\varphi\vee\psi,\textrm{~if~}\acc\agn{:}\varphi \textrm{~or~}\acc\agn{:}\psi$

\item $\rej\agn{:}\varphi\vee\psi,\textrm{~if~}\rej\agn{:}\varphi \textrm{~and~}\rej\agn{:}\psi$

\item $\nacc\agn{:}\varphi\vee\psi,\textrm{~if~}\nacc\agn{:}\varphi \textrm{~and~}\nacc\agn{:}\psi$

\item $\nrej\agn{:}\varphi\vee\psi,\textrm{~if~}\nrej\agn{:}\varphi \textrm{~or~}\nrej\agn{:}\psi$
\end{enumerate}

Notice that the definition of the cognitive attitude $\acc$ of acceptance in terms of truth-values coincides with the bifilter $\mathcal{F}$ of $FOUR$ (cf.\ Figure~\ref{fig:1}). We thus have that $\mathfrak{M}=\langle\mathbf{4},\acc,\mathcal{O}\rangle$. Hence, just the cognitive attitude $\acc$ and its complement $\mathbf{4}-\acc=\nacc$ determine the notion of validity of $\mathbf{E}$. Let $C\agn{:}\Phi=\{C\agn{:}\varphi| \textrm{ for all } \varphi\in\Phi\}$, where $C\in COG$ and $\Phi\subseteq\mathcal{S}$. The society of agents $SOC$ of $\mathbf{E}$ is the set of agents based on $\mathfrak{M}$. For all $\alpha\in\mathnormal{S}$ and all $\Gamma\subseteq\mathnormal{S}$:

\[\Gamma\vDash^4\alpha\textrm{ iff there is no }\agn\in SOC\textrm{ such that }\acc\agn{:}\Gamma\textrm{ and }\nacc\agn{:}\alpha.\]

\noindent
According to the relation $\vDash^4$, we say that $\alpha$ is a consequence of $\Gamma$ if no agent of the society $SOC$ accepts all of the premises while not-accepting the conclusion.

Considering also the cognitive attitude of rejection, we notice that $\rej$ in terms of truth-values is the (prime) bifilter of the bilattice ${FOUR^-}$, formed by the informational order $\leq_i$ and the inverse of the logical order $\leq_t^-$ (cf. Figure \ref{fig:2}).

\begin{figure}[ht]
\begin{center}
\ifx\du\undefined
  \newlength{\du}
\fi
\setlength{\du}{15\unitlength}
\begin{tikzpicture}
\pgftransformxscale{.85000000}
\pgftransformyscale{-.85000000}
\draw (25.000000\du,5.000000\du)--(27.000000\du,7.000000\du)--(25.000000\du,9.000000\du)--(23.000000\du,7.000000\du)--cycle;
\draw (25.000000\du,5.000000\du)--(27.000000\du,7.000000\du)--(25.000000\du,9.000000\du)--(23.000000\du,7.000000\du)--cycle;
\node[anchor=west] at (27.000000\du,7.000000\du){$\ff$};
\node[anchor=west] at (22.00000\du,7.000000\du){$\ttt$};
\node[anchor=west] at (24.3500000\du,4.500000\du){$\top$};
\node[anchor=west] at (24.3500000\du,9.500000\du){$\bot$};
{
\pgfsetarrowsstart{to}
\draw (22.000000\du,4.000000\du)--(22.000000\du,10.000000\du);
}
{
\pgfsetarrowsend{to}
\draw (22.000000\du,10.000000\du)--(28.000000\du,10.000000\du);
}
\node[anchor=west] at (20.50000\du,4.500000\du){$\leq_i$};
\node[anchor=west] at (27.00000\du,10.60000\du){$\leq_t^-$};
\end{tikzpicture}
\end{center}
	\vspace{-.5cm}\caption{ $\mathit{FOUR^-}$}
	\label{fig:2}
\end{figure}

From $FOUR^-$ we can define the logic $\mathbf{E}^{-}=\langle\mathnormal{S},\vDash^{4-}\rangle$, in which the notion of validity is defined in terms of the cognitive attitudes of rejection $\rej$ and its complement, not-rejection $\nrej$. The society of agents $SOC^-$ of $\mathbf{E}^-$ is the set of agents based on the matrix $\mathfrak{M}^-=\langle\mathbf{4},\rej,\mathcal{O}\rangle$. For all $\alpha\in\mathnormal{S}$ and all $\Gamma\subseteq\mathnormal{S}$, 
\[\Gamma\vDash^{4-}\alpha\textrm{ iff there is no }\agn\in SOC^- \textrm{ such that }\nrej\agn{:}\Gamma\textrm{ and }\rej\agn{:}\alpha. \]

\noindent
According to the relation $\vDash^{4-}$, we say that $\alpha$ is a consequence of $\Gamma$ if no agent of the surveyed society not-rejects all of the premises while rejecting the conclusion.

As a way of securing that the rejected information content should have the same importance as the accepted informational content in a logic based in $FOUR$, we will adopt an alternative semantic matrix created in \cite{Malinowski90a} and generalized in \cite{ShramkoWansing2011}. The so-called symmetric matrix has two sets of designated truth-values which correspond the cognitive attitudes of acceptance and rejection.

Consider the symmetric matrix
 \[\mathfrak{M}^*=\langle\mathbf{4},~\acc,~\rej,~\mathcal{O}\rangle,\] 
 
\noindent such that, $\mathbf{4}=\{\ff, \bot, \top, \ttt\},~\acc=\{\top,\ttt\},~\rej=\{\ff, \top\}$ and $~\mathcal{O}=\{\sqcap_t,  \sqcup_t, -_t\}$. The acceptance and the rejection of informational content by a certain agent is represented in $\mathfrak{M}^*$ respectively by $\acc$ and $\rej$. Notice that $\top\in\acc\cap\rej$ and $\bot\in\mathbf{4}-\acc\cup\rej$ (cf. Figure \ref{fig:mg}), thus, it is possible that the informational content of a statement is at once both accepted and rejected, or at once neither accepted nor rejected, by an agent.

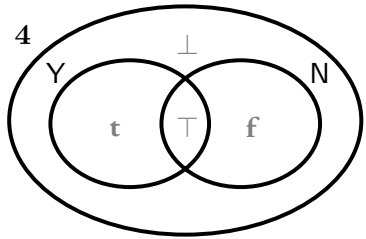
\begin{figure}[h!] 
\begin{center}
\ifx\du\undefined
  \newlength{\du}
\fi
\setlength{\du}{15\unitlength}
\begin{tikzpicture}
\pgftransformxscale{.8000000}
\pgftransformyscale{-.8000000}
\pgfsetlinewidth{0.100000\du}
\pgfpathellipse{\pgfpoint{8.250000\du}{7.900000\du}}{\pgfpoint{5.550000\du}{0\du}}{\pgfpoint{0\du}{3.600000\du}}
\pgfsetlinewidth{0.100000\du}
\pgfpathellipse{\pgfpoint{6.500000\du}{8.000000\du}}{\pgfpoint{2.500000\du}{0\du}}{\pgfpoint{0\du}{2.000000\du}}
\pgfsetlinewidth{0.100000\du}
\pgfpathellipse{\pgfpoint{10.000000\du}{8.000000\du}}{\pgfpoint{2.500000\du}{0\du}}{\pgfpoint{0\du}{2.000000\du}}
\pgfusepath{stroke}
\node[anchor=west] at (5.500000\du,8.150000\du){\textcolor{gray}{$\ttt$}};
\node[anchor=west] at (9.800000\du,8.100000\du){\textcolor{gray}{$\ff$}};
\node[anchor=west] at (2.5000000\du,5.20000\du){$\mathbf{4}$};
\node[anchor=west] at (7.60000\du,5.500000\du){\textcolor{gray}{$\bot$}};
\node[anchor=west] at (7.60000\du,8.100000\du){\textcolor{gray}{$\top$}};
\node[anchor=west] at (3.5000000\du,6.400000\du){$\acc$};
\node[anchor=west] at (11.80000\du,6.400000\du){$\rej$};
\end{tikzpicture}
\end{center}
\vspace{-.5cm}
\caption{\emph{Symmetric matrix $\mathfrak{M}^*$.}}
\label{fig:mg}
\end{figure}

Compared to the matrices $\mathfrak{M}$ and $\mathfrak{M}^-$, the matrix $\mathfrak{M}^*$ may be associated with new definitions of entailment, and therefore, to different multi-valued logics related to $FOUR$, some of which are presented in the next section.

\section{How a computer might think}

Given the matrix $\mathfrak{M}^*=\langle\mathbf{4},~\acc,~\rej,~\mathcal{O}\rangle$, we define a logic which has two Tarskian consequence relations. Such logic follows the proposal of \cite{ShramkoWansing2007,ShramkoWansing2011} called $k$-dimensional Tarskian logic. A $k$-dimensional Tarskian logic has $k$ independent Tarskian consequence relations, where $k{\in}\mathbb{N}$ and $k{>}1$. A consequence relation is said to be independent of the others when it cannot be defined from other consequence relations present in the $k$-dimensional logic.

Accordingly, we define the 2-dimensional logic $\mathbf{E}^{2D}$ associated with $\mathfrak{M}^*$ as:
\[\mathbf{E}^{2D}=\langle\mathnormal{S},\vDash^{t},\vDash^{f}\rangle,\] 

\noindent where $\mathnormal{S}$ is the same language as $\mathbf{E}$, $SOC^*$ is the society of agents based on $\mathfrak{M}^*$ and the consequence relations are defined as $\vDash^{t}\;=\;\vDash^{4}$ and $\vDash^{f}\;=\;\vDash^{4-}$.

The logic $\mathbf{E}^{2D}$ contains a Tarskian consequence relation dealing with the acceptance (or not) of informational content and another Tarskian consequence relation dealing with the rejection (or not) of informational content. This logic seems to be one step closer to our goal, but it is not clear how it connects acceptance and rejection of statements.

The Tarskian consequence relations, however, are not the only definitions of consequence relations that may be associated to the matrix $\mathfrak{M}^*$. Shramko and Wansing \cite{ShramkoWansing2007,ShramkoWansing2011} notice also the possibility of defining other consequence relations for logics associated with the matrices with more than one set of designated values, such as the symmetric matrix:

\begin{quote}
\footnotesize{It appears that the mere multiplication of semantical values not only affords room for defining various entailment relations, but this also allows one to define semantical relations differing in important respects from the familiar notion of semantic consequence (cf. \cite[p.~208]{ShramkoWansing2011}.)}
\end{quote}

One of the definitions of semantic consequence which may be associated with $\mathfrak{M}^*$ is the \emph{$q$-entailment} ---$q$ stands for \emph{quasi}---, proposed by Malinowski together with the $q$-matrix \cite{Malinowski90a}. The reasoning given by \emph{$q$-entailment} is related to reasoning by hypothesis, widely adopted in the empirical sciences. An inference of \emph{$q$-entailment} is valid when the conclusion is \emph{accepted} whenever all the premises are \emph{not-rejected}, and so, in the case where the conclusion is not-accepted, some of the premises ought to be rejected. Accordingly, for all $\Gamma\subseteq\mathnormal{S}$ and $\alpha\in\mathnormal{S}$,
\[\Gamma\vDash^q\alpha \textrm{ iff there is no agent } \agn\in SOC^* \textrm{ such that } \nrej\agn{:}\Gamma \textrm{ and } \nacc\agn{:}\alpha.\]

Among the properties of the Tarskian notion of consequence, Reflexivity is not guaranteed by \emph{$q$-entail\-ment}, because some information may simultaneously be not-accepted and not-rejected by agents of a given society.

The $q$-entailment associated with $\mathfrak{M}^*$ also does not respect the property of Transitivity.%
\footnote{
Originally \emph{$q$-entailment} is a transitive relation, but not reflexive, because it is based on a semantic matrix where the set of accepted values and the set of rejected values are disjoint ($\acc\cap\rej=\emptyset$). 
} %
Indeed, suppose that there is an agent $\agn$ of the society such that  $\nrej\agn{:}\Gamma$, $\rej\agn{:}\Delta$ and $\nacc\agn{:}\alpha$. We have that $\Gamma\vDash^q\delta\textrm{, for all}~\delta\in\Delta$ and $\Delta\vDash^q\alpha$. Since there is an agent $\agn$ which not-rejects the informational content of the propositions in $\Gamma$ and not-accepts the informational content of~$\alpha$, by definition we have that $\Gamma\not\vDash^q\alpha$.

Another non-Tarskian consequence relation that may be defined based in $\mathfrak{M}^*$ is the $p$-\emph{entailment}---$p$ stands for \emph{plausible}--- cf. \cite{Frankowski04a}. The reasoning expressed by \emph{$p$-entailment} allows for a decrease of certainty from the premises to the conclusion. An inference in the form of $p$-{entailment} is valid when, all premises being accepted, the conclusion is not-rejected. Accordingly, for all $\Gamma\subseteq\mathnormal{S}$ and all $\alpha\in\mathnormal{S}$,
\[\Gamma\vDash^p\alpha \textrm{ iff there is no agent } \agn\in SOC^*\textrm{ such that } \acc\agn{:}\Gamma \textrm{ and }\rej\agn{:}\alpha.\]

The $p$-{entailment} associated with $\mathfrak{M}^*$ does not respect the properties of Reflevixity and Transitivity. Reflexivity fails considering a situation in which $\alpha$ is accepted and rejected by an agent of a given society.
\footnote{Originally $p$-{entailment} is reflexive, but not transitive, because it is based in a semantic matrix whose truth-values belong either to the set of accepted values or to the set of rejected values ($\acc\cup\rej=\mathcal{V}$).
}
Assume that $\Gamma\vDash^p\delta\textrm{ for all}~\delta\in\Delta$ and assume as well that $\Delta\vDash^p\alpha$. Suppose that there is an agent $\agn\in SOC^*$ such that $\acc\agn{:}\Gamma$, $\nacc\agn{:}\Delta$ and $\rej\agn{:}\alpha$. Then, $\Gamma\not\vDash^p\alpha$, since $\agn$ accepts the informational content of the propositions in $\Gamma$ and rejects the informational content of the proposition~$\alpha$.



Considering the semantic consequence relations of $q$- and $p$-{entailment}, we define the logic $\mathbf{E}^4$ based on $\mathfrak{M}^*$ with four entailment relations, following \cite{ShramkoWansing2007,ShramkoWansing2011}:
\[\mathbf{E}^4=\langle\mathnormal{S},\vDash^t,\vDash^f,\vDash^q,\vDash^p\rangle.\]

According to Shramko and Wansing (\cite{ShramkoWansing2007}, p. 140), these four consequence relations may be represented in a bilattice such as the truth-values of $\mathbf{4}$. 
This structure of entailment relations lead the authors to consider, without going into much detail, that each of the consequence relations seems to represent a truth-value of~$\mathbf{4}$. 
Such idea, however, is deemed problematic by the authors themselves, once the four consequence relations are not independent and can be reduced to only two, according to the following result:

\begin{prop}[\cite{ShramkoWansing2007}, pp. 137-8] Given the symmetric matrix $\mathfrak{M}^*$, if we introduce in its structure the four entailment relations previously defined, we have that $\vDash^t\;=\;\vDash^f$ and $\vDash^q\;=\;\vDash^p$. Moreover, $\vDash^t$ (or $\vDash^f$) $\;\neq\;$ $\vDash^p$ (or  $\vDash^q$).
\end{prop}

\begin{proof}
The fact that $\vDash^t\;=\;\vDash^f$ is proved in \cite[Proposition 4]{Dunn2000}. In terms of a society of agents, for each agent $\agn$, an agent $\agn^*$ is defined such that  $\nrej\agn^*{:}\varphi$ iff $\acc\agn{:}\varphi$; $\acc\agn^*{:}\varphi$ iff $\nrej\agn{:}\varphi$; $\nacc\agn^*{:}\varphi$ iff $\rej\agn{:}\varphi$; and $\rej\agn^*{:}\varphi$ iff $\nacc\agn{:}\varphi$. Assume that $\Gamma\vDash^t\psi$ and consider an agent $\agn$ such that $\rej\agn{:}\psi$. Then, $\nacc\agn^*{:}\psi$. This way, 
we obtain $\nacc\agn^*{:}\Gamma$, and thus, $\rej\agn{:}\Gamma$, therefore $\Gamma\vDash^f\psi$. The proof of the converse result is analogous.

The fact that $\vDash^q\;=\;\vDash^p$ is proved in \cite[Proposition~1]{ShramkoWansing2007} and is verified in the following way:

\noindent
[$\vDash^p\subseteq\vDash^q$] Suppose that $\Delta\not\vDash^q\alpha$. By definition, there is an agent $\agn$ such that $\nacc\agn{:}\Delta$ and $\nrej\agn{:}\alpha$. We want to show that $\Delta\not\vDash^p\alpha$, that is, that there is an agent $\agn'$ such that $\acc\agn'{:}\Delta$ and $\rej\agn'{:}\alpha$. Take $\agn'=\agn^*$. Notice that $\acc\agn^*{:}\Delta$ and $\rej\agn^*{:}\alpha$. Therefore, $\Delta\not\vDash^p\alpha$.

\noindent
[$\vDash^q\subseteq\vDash^p$] Suppose that $\Delta\not\vDash^p\alpha$. By definition, there is an agent $\agn$ such that $\acc\agn{:}\Delta$ and $\rej\agn{:}\alpha$. We want to show that $\Delta\not\vDash^q\alpha$, that is, that there is an agent $\agn'$, such that $\nacc\agn'{:}\Delta$ and $\nrej\agn'{:}\alpha$. Take $\agn'=\agn^*$. Notice that $\nrej\agn^*{:}\Delta$ and $\nacc\agn^*{:}\alpha$. Therefore, $\Delta\not\vDash^q\alpha$.

The definition of $\vDash^q$ (or, equally, of $\vDash^p$) does not coincide with $\vDash^t$ (or, equally, with $\vDash^f$) when they are associated with $\mathfrak{M}^*$, because the former is not Tarskian, being neither reflexive nor symmetric, and the latter is Tarskian.
\end{proof}


A logic which intends to deal with informational content even in the presence of inconsistency or partiality should not focus only on statements ``said to be true'' in detriment of statements ``said to be false'' as determinants of the notion of validity of an inference. Beyond the reasoning based on acceptance, different kinds of reasoning can emerge from accepted and rejected statements, such as the hypothetical reasoning expressed by {$q$-entailment}, the pragmatic reasoning of {$p$-entailment} and the reasoning based on the preservation of that which is not rejected.

The first step taken in order to define a logic which deals with informational content is 
to change some definitions concerning the notion of consequence. 
Instead of epistemic situations, which coincide with the truth-values, we give way to the cognitive attitudes of acceptance and rejection of a given informational content by a set of agents as primitive objects that define inferences. Following this path, we adopt the symmetric matrix which allows for the expression of logical reasoning including those which give origin to the definition of non-Tarskian consequence relations.

The matrix $\mathfrak{M}^*$ allows for the definition of entailment relations not only in terms of preservation of acceptance and rejection of the informational content---characteristic of the Tarskian consequence relations---but also in terms of the interaction among the cognitive attitudes of acceptance and rejection of the informational content of statements.

We present, in the following, a logic whose associated semantic consequence relation, `{$B$-entailment}', can effectively express the diverse types of reasoning involving accepted and rejected informational content. This consequence relation generalizes in one single structure the consequence relations of $\vDash_t$, $\vDash_f$, $\vDash_q$ and $\vDash_p$.

\section{The logic $\mathbf{E}^B$}

Given a semantics based on a society of agents which entertain the cognitive attitudes of acceptance or rejection towards a given informational content of consulted statements, we will adopt a semantic consequence relation, called \emph{$B$-entailment}, to define the logic $\mathbf{E}^B$. The definition of {$B$-entailment} is able to cover all aspects related to acceptance and rejection, including, in particular, reasoning expressed by the previously defined four notions of consequence, to such an extent that it does not confound $\vDash^t$ with $\vDash^f$ and $\vDash^q$ with $\vDash^p$. We call $B$-logic the logic associated with a {$B$-entailment}. 

Let $\mathbf{E}^B$ be the $B$-logic
\[\mathbf{E}^B=\langle\mathcal{S},\Bent{\textcolor{white}{\cdot}}{}{\textcolor{white}{\cdot}}{}\rangle\]

\noindent where $\mathcal{S}$ is the same language of $\mathbf{E}$ and  $\Bent{\textcolor{white}{\cdot}}{}{\textcolor{white}{\cdot}}{}\subseteq\wp(\mathcal{S})^4$ is a {$B$-{entailment}}, defined as follows.  Given a set of agents $SOC^*$ based on the matrix $\mathfrak{M}^*$, for all $\Gamma, \Delta, \Phi, \Psi\subseteq\mathcal{S}$, 

\[\Bent{\Gamma}{\Psi}{\Phi}{\Delta}\textrm{~iff~there is no }\agn\in SOC^* \textrm{ such that }\acc s{:}\Gamma\textrm{~and~}\nacc s{:}\Delta\textrm{~and~}\rej s{:}\Phi\textrm{~and~}\nrej s{:}\Psi.\]

\noindent that is, an inference $\Bent{\Gamma}{\Psi}{\Phi}{\Delta}$ of $\mathbf{E}^B$ is valid if there is no agent $s\in SOC^*$ such that $s$ accepts all sentences of $\Gamma$, not-accepts all the sentences of $\Delta$, rejects all sentences of $\Phi$ and not-rejects all the sentences of $\Psi$.
In case 
$\textrm{there is an } \agn\in SOC^* \textrm{ such that }\acc \agn{:}\Gamma\textrm{~and~}\nacc \agn{:}\Delta\textrm{~and~}\rej \agn{:}\Phi\textrm{~and~}\nrej \agn{:}\Psi$ 
we say we are dealing with an invalid inference in the form of $B$-{entailment}, and denote it by writing
$\nBent{\Gamma}{\Psi}{\Phi}{\Delta}$.

$B$-{entailment} can express different types of reasoning related to acceptance and to rejection such as the ones given by the consequence relations $\vDash_t$, $\vDash_f$, $\vDash_q$ and $\vDash_p$, as shown by the Table \ref{tab:tfpqB} below. 

\begin{table}[ht]
\begin{center}
\small
$\begin{array}{c@{\textrm{ ~iff~ }}c@{\textrm{ ~iff~ }}c}\vspace{.2cm}

\Gamma\vDash^t\delta&\textrm{ there is no agent }\agn\textrm{ such that }\acc_{\agn}{:}\Gamma\textrm{ and }\nacc_{\agn}{:}\delta&\Bent{\Gamma}{}{}{\delta}\\\vspace{.2cm}

\Psi\vDash^f\varphi&\textrm{ there is no agent }\agn\textrm{ such that }\nrej_{\agn}{:}\Psi\textrm{ and }\rej_{\agn}{:}\varphi&\Bent{}{\Psi}{\varphi}{}\\\vspace{.2cm}

\Psi\vDash^q\delta&\textrm{ there is no agent }\agn\textrm{ such that }\nrej_{\agn}{:}\Psi\textrm{ and }\nacc_{\agn}{:}\delta&\Bent{}{\Psi}{}{\delta}\\\vspace{.2cm}

\Gamma\vDash^p\varphi&\textrm{ there is no agent }\agn\textrm{ such that }\acc_{\agn}{:}\Gamma\textrm{ and }\rej_{\agn}{:}\varphi&\Bent{\Gamma}{}{\varphi}{}

\end{array}$
\normalsize
\end{center}
\vspace{-.4cm}
 \caption{Some types of reasoning expressed by using {$B$-entailment}, where $\Gamma,\Psi\subseteq\mathnormal{S}$ and $\varphi,\delta\in\mathnormal{S}$.}
 \label{tab:tfpqB}
 \end{table}

Even though $B$-{entailment}\footnote{\cite{Bochman98} defines a consequence relation similar to $B$-{entailment}, called \emph{biconsequence}.} is not a Tarskian consequence relation, its two dimensions allow us to observe properties related to Reflexivity, Monotonicity and Transitivity.
 
\begin{prop}
For all $\alpha\in\mathcal{S}$ and $\Gamma,\Delta, \Phi, \Psi$, $\Gamma',\Delta', \Phi', \Psi'~{\subseteq}~\mathcal{S}$, {$B$-entailment} respects the following properties: 

 \begin{description}
\item[\textbf{$B$-reflexivities}]\textcolor{white}{ }

\begin{itemize}

\item[$(t)$] $\Bent{\alpha}{}{}{\alpha}$ 

\item[$(f)$] $\Bent{}{\alpha}{\alpha}{}$ 
\end{itemize}

\item[\textbf{$B$-monotonicity}]\textcolor{white}{ } 

\begin{itemize}
\item[ ] If $\Bent{\Gamma}{\Psi}{\Phi}{\Delta}$, then $\Bent{\Gamma,\Gamma'}{\Psi,\Psi'}{\Phi,\Phi'}{\Delta,\Delta'}$ 
\end{itemize}

\item[\textbf{$B$-transitivities}]\textcolor{white}{ }

\begin{itemize}
\item[$(t)$] If $\Bent{\alpha,\Gamma}{\Psi}{\Phi}{\Delta}$ and $\Bent{\Gamma}{\Psi}{\Phi}{\Delta,\alpha}$, then $\Bent{\Gamma}{\Psi}{\Phi}{\Delta}$

\item[$(f)$]  If $\Bent{\Gamma}{\alpha,\Psi}{\Phi}{\Delta}$ and $\Bent{\Gamma}{\Psi}{\Phi,\alpha}{\Delta}$, then $\Bent{\Gamma}{\Psi}{\Phi}{\Delta}$
\end{itemize}
\end{description}
\end{prop}

\begin{proof}
$B$-reflexivity ($t$). Suppose that $\nBent{\alpha}{}{}{\alpha}$. Then there is an agent $\agn \in SOC^*$ such that $\acc\agn{:}\alpha$ and $\nacc\agn{:}\alpha$. This means that $\agn(\alpha)\in\acc\cap\nacc=\emptyset$, which is absurd.

$B$-reflexivity ($f$). Suppose that$\nBent{}{\alpha}{\alpha}{}$. Then there is an agent $\agn \in SOC^*$ such that $\nrej\agn{:}\alpha$ and $\rej\agn{:}\alpha$. This means that $\agn(\alpha)\in\nrej\cap\rej=\emptyset$, which is absurd.

$B$-monotonicity. Suppose, by contraposition, that $\nBent{\Gamma',\Gamma''}{\Psi',\Psi''}{\Phi',\Phi''}{\Delta',\Delta''}$ . Then, there is an $\agn \in SOC^*$ such that $\acc\agn{:}\Gamma'$, $\acc\agn{:}\Gamma''$, $\nacc\agn{:}\Delta'$, $\nacc\agn{:}\Delta''$, $\rej\agn{:}\Phi'$, $\rej\agn{:}\Phi''$, $\nrej\agn{:}\Psi'$ and $\nrej\agn{:}\Psi''$. Notice that there is no $\agn \in SOC^*$ such that $\acc\agn{:}\Gamma'$, $\nacc\agn{:}\Delta'$, $\rej\agn{:}\Phi'$ and $\nrej\agn{:}\Psi'$. Therefore, $\nBent{\Gamma'}{\Psi'}{\Phi'}{\Delta'}$.

$B$-transitivity ($t$). 
Let $\Bent{\alpha,\Gamma}{\Psi}{\Phi}{\Delta}$ and $\Bent{\Gamma}{\Psi}{\Phi}{\Delta,\alpha}$. Then, there is no agent $\agn\in SOC^*$ such that $\acc\agn{:}\Gamma$, $\acc\agn{:}\alpha$, $\nacc\agn{:}\Delta$, $\rej\agn{:}\Phi$, $\nrej\agn{:}\Psi$. There is also no agent $\agn\in SOC^*$ such that $\acc\agn{:}\Gamma$, $\nacc\agn{:}\Delta$, $\nacc\agn{:}\alpha$, $\rej\agn{:}\Phi$, $\nrej\agn{:}\Psi$.
Therefore, we have that there is no $\agn\in SOC^*$, such that $\acc\agn{:}\Gamma$, $\nacc\agn{:}\Delta$, $\rej\agn{:}\Phi$, $\nrej\agn{:}\Psi$, and so $\Bent{\Gamma}{\Psi}{\Phi}{\Delta}$.

$B$-transitivity ($f$).  The proof is similar to that of $B$-transitivity ($t$).
\end{proof}

The two aspects of Reflexivity valid in $\mathbf{E}^B$ correspond to the forms of reasoning of the Tarskian consequence relations $\vDash^t$ and $\vDash^f$.
By contrast, there are other forms of reasoning, which express other aspects of Reflexivity, that are not valid in {$B$-entailment}, such as $\Bent{\alpha}{}{\alpha}{}$ and $\Bent{}{\alpha}{}{\alpha}$.

\begin{prop} \begin{enumerate} \item$\nBent{\alpha}{}{\alpha}{}$  
\item$\nBent{}{\alpha}{}{\alpha}$
\end{enumerate} \label{prop:refpq}
\end{prop}
\begin{proof}
\textit{1.} Let $s(\alpha)=\top$. Then, there is $\agn\in SOC^*$ such that $\nacc\agn{:}\alpha$ and $\nrej\agn{:}\alpha$. Thus, $\nBent{\alpha}{}{\alpha}{}$.

\textit{2.} Let $s(\alpha)=\bot$. Then, there is $\agn\in SOC^*$ such that $\acc\agn{:}\alpha$ and $\rej\agn{:}\alpha$. Thus, $\nBent{}{\alpha}{}{\alpha}$.
\end{proof}
 
 The non-valid aspects of Reflexivity coincide with the forms of $\vDash^q$ and of ~$\vDash^p$ in which Reflexivity is also not valid when associated with the matrix  $\mathfrak{B}^{\mathbf{E}}$. In addition, when $\Bent{\alpha}{}{\alpha}{}$ is valid, the semantics ``loses'' the truth-value $\top$ and when $\Bent{}{\alpha}{}{\alpha}$ is valid, the semantics ``loses'' the truth-value $\bot$. 
 
\newpage
 Let us now examine some examples of valid and invalid inferences of $\mathbf{E}$, comparing them to $\mathbf{E}^B$.

\begin{description}
\item[\textbf{Conjunction introduction:}] $\alpha,~\beta\vDash^4\alpha\wedge\beta$

In this case, there are eight ways of expressing the introduction of conjunction using {$B$-entailment}, but only $\Bent{\alpha,~\beta}{}{}{\alpha\wedge\beta}$ and $\Bent{}{\alpha,~\beta}{\alpha\wedge\beta}{}$ are valid. These inferences express, respectively, the reasoning of the standard consequence relations $\vDash^t$ and $\vDash^f$ (cf. Table \ref{tab:tfpqB}). 

The following forms, for example, are not valid in $\mathbf{E}^B$:

 $\Bent{}{\alpha,~\beta}{}{\alpha\wedge\beta}$. Consider an agent $\agn' \in SOC^*$ such that $\nacc\agn{:}\alpha$, $\nrej\agn'{:}\alpha$ and $\nrej\agn'{:}\beta$ (that is, $\agn'(\alpha)=\top$ and $\agn'(\beta)\in\{\top,\ttt\}$). The inference is invalid since $\nrej\agn'{:}\alpha$ and $\nrej\agn'{:}\beta$ and $\nacc\agn'{:}\alpha\wedge\beta$ (by the recursive clause 2.7).

$\Bent{\beta}{\alpha}{}{\alpha\wedge\beta}$. Consider an agent $\agn''\in SOC^*$ such that $\nacc\agn''{:}\alpha$ and $\nrej\agn''{:}\alpha$ (that is, $\agn''(\alpha)=\bot$) and $\acc\agn''{:}\beta$ (that is $\agn''(\beta)\in\{\top, \ttt\}$). The inference is invalid since $\acc\agn''{:}\beta$, $\nrej\agn''{:}\alpha$ and $\nacc\agn''{:}\alpha\wedge\beta$ (by clause 2.7).

Countermodels for $\nBent{\alpha,~\beta}{}{\alpha\wedge\beta}{}$,  $\nBent{\beta}{\alpha}{\alpha\wedge\beta}{}$,  $\nBent{\alpha}{\beta}{\alpha\wedge\beta}{}$ and  $\nBent{\alpha}{\beta}{}{\alpha\wedge\beta}$ may be described in a similar way.

\item[\textbf{Principle of explosion:}] $\alpha\wedge\neg\alpha\not\vDash^4\beta$

The principle of explosion is always invalid in any aspect of the $B$-{entailment}:
 $\nBent{\alpha\wedge\neg\alpha}{}{}{\beta}$, $\nBent{}{\alpha\wedge\neg\alpha}{\beta}{}$,$\nBent{}{\alpha\wedge\neg\alpha}{}{\beta}$ e $\nBent{\alpha\wedge\neg\alpha}{}{\beta}{}$. 
We will show a countermodel for 
$\Bent{\alpha\wedge\neg\alpha}{}{}{\beta}$. Consider an agent $\agn$, such that $\acc\agn{:}\alpha$, $\rej\agn{:}\alpha$ (that is, $\agn(\alpha)=\top$) and $\nacc\agn{:}\beta$. By clause 2.1, $\acc\agn{:}\alpha$ and $\acc\agn{:}\neg\alpha$ and by clause 2.5, $\acc\agn{:}\alpha\wedge\neg\alpha$. Thus, there is an agent $\agn$, such that $\acc\agn{:}\alpha\wedge\neg\alpha$ and $\nacc\agn{:}\beta$. Therefore, $\nBent{\alpha\wedge\neg\alpha}{}{}{\beta}$.
Notice that the principle of explosion being invalid has to do with the fact that $\mathbf{E}^B$ is a relevant logic. 

\item[\textbf{Principle of the excluded middle:}] $\beta\not\vDash^4\neg\alpha\vee\alpha$

The principle of the excluded middle is also always invalid in any aspect of the $B$-{entailment}: $\nBent{\beta}{}{}{\neg\alpha\vee\alpha}$, $\nBent{}{\beta}{\neg\alpha\vee\alpha}{}$, $\nBent{}{\beta}{}{\neg\alpha\vee\alpha}$ e $\nBent{\beta}{}{\neg\alpha\vee\alpha}{}$. 
We will show a countermodel for 
$\Bent{}{\beta}{}{\neg\alpha\vee\alpha}$. Consider an agent $\agn$, such that $\nacc\agn{:}\alpha$ and $\nrej\agn{:}\alpha$ (that is, $\agn(\alpha)=\bot$). By clause 2.3, $\nacc\agn{:}\alpha$ and $\nacc\agn{:}\neg\alpha$, and by clause 2.6 $\nacc\agn{:}\neg\alpha\vee\alpha$. Thus, there is an agent $\agn$, such that $\nrej\agn{:}\beta$ and $\nacc\agn{:}\neg\alpha\vee\alpha$. Therefore, $\nBent{}{\beta}{}{\neg\alpha\vee\alpha}$.
Notice that the principle of the excluded middle being invalid has to do with the fact that $\mathbf{E}^B$ is a relevant logic. 

\end{description}  

In the following, we propose a sequent calculus and prove the characterization results of $\mathbf{E}^B$. 

\paragraph{\textbf{Sequent calculus for $\mathbf{E}^B$}}\textcolor{white}{.}

We call $B$-sequents the expressions of the form $\Bcon{\Gamma}{\Psi}{\Phi}{\Delta}$ in which $\Bcon{\textcolor{white}{\cdot}}{}{\textcolor{white}{\cdot}}{}$ is a sequent symbol with four positions and $\Gamma, \Psi, \Delta, \Phi \subseteq \mathcal{S}$ are finite sets of formulas of the language. Let $\Lambda=\{\lambda_1, ..., \lambda_n\}, n\in\mathbb{N}$ and let $\bigvee C\agn{:}\Lambda$ be a shorthand for ``the agent $\agn$ entertains the cognitive attitude $C$ towards $\lambda_1$ or ... or to $\lambda_n$''.
For all $\agn\in SOC^*$, the meaning of the $B$-sequent $\Bcon{\Gamma}{\Psi}{\Phi}{\Delta}$ is given by:
\[{\bigvee}\nacc\agn{:}\Gamma\textrm{ or }{\bigvee}\acc\agn{:}\Delta\textrm{ or }{\bigvee}\nrej\agn{:}\Phi\textrm{ or }{\bigvee}\rej\agn{:}\Psi.\]

For 
any formulas $\alpha,\beta \in \mathcal{S}$ and any finite sets of formulas $\Gamma,\Delta, \Phi, \Psi$, $\Gamma',\Delta', \Phi', \Psi'~{\subseteq}~\mathcal{S}$, the system of $B$-sequents for $\mathbf{E}^B$ contains the following rules:
\paragraph{\textbf{Structural rules}}\hspace{-.4cm}:
\\

$B$-initial sequents
\\

\quad
\infer[in_t]{\Bcon{\alpha}{}{}{\alpha}}{}
\qquad
\infer[in_f]{\Bcon{}{\alpha}{\alpha}{}}{}
\\

$B$-weakening
\\

\quad
\infer[weak]{\Bcon{\Gamma',\Gamma}{\Psi',\Psi}{\Phi,\Phi'}{\Delta,\Delta'}}{\Bcon{\Gamma}{\Psi}{\Phi}{\Delta}}
\\

$B$-cuts
\\

\quad
\infer[cut_t]{\Bcon{\Gamma}{\Psi}{\Phi}{\Delta}}{\Bcon{\alpha,\Gamma}{\Psi}{\Phi}{\Delta} \quad \Bcon{\Gamma}{\Psi}{\Phi}{\Delta,\alpha}}
\quad
\infer[cut_f]{\Bcon{\Gamma}{\Psi}{\Phi}{\Delta}}{\Bcon{\Gamma}{\alpha,\Psi}{\Phi}{\Delta} \quad \Bcon{\Gamma}{\Psi}{\Phi,\alpha}{\Delta}} 
\paragraph{\textbf{Logical rules}}\hspace{-.4cm}:
\\

Conjunction
\\

\quad
\infer[\Rightarrow_{\ff} \wedge]{\BconCONT{}{}{}{, \alpha\wedge\beta}}{\BconCONT{}{}{}{, \alpha} \quad \BconCONT{}{}{}{, \beta}}
\qquad
\infer[\Rightarrow_{\ttt} \wedge]{\BconCONT{}{}{, \alpha\wedge\beta}{}}{\BconCONT{}{}{, \alpha}{} \quad \BconCONT{}{}{, \beta}{}}
\\

\quad
 \infer[\wedge \Rightarrow_{\ff}]{\BconCONT{\alpha\wedge\beta, }{}{}{}}{\BconCONT{\alpha,\beta, }{}{}{}}
\qquad\qquad\infer[\wedge \Rightarrow_{\ttt}]{\BconCONT{}{\alpha\wedge\beta, }{}{}}{\BconCONT{}{\alpha,\beta, }{}{}}
\\

Disjunction
\\

\quad
\infer[\vee \Rightarrow_{\ttt}]{\BconCONT{}{\alpha\vee\beta, }{}{}}{\BconCONT{}{\alpha, }{}{} \quad \BconCONT{}{\beta, }{}{}}
\qquad\infer[\vee \Rightarrow_{\ff}]{\BconCONT{\alpha\vee\beta, }{}{}{}}{\BconCONT{\alpha, }{}{}{} \quad \BconCONT{\beta, }{}{}{}}
\\

\quad\infer[\Rightarrow_{\ttt} \vee]{\BconCONT{}{}{, \alpha\vee\beta}{}}{\BconCONT{}{}{, \alpha,\beta}{}}
\qquad\qquad\infer[\Rightarrow_{\ff} \vee]{\BconCONT{}{}{}{, \alpha\vee\beta}}{\BconCONT{}{}{}{, \alpha,\beta}}
 \\
 
Negation
\\

\quad
\infer[\neg \Rightarrow_{\ff}]{\BconCONT{\neg\alpha, }{}{}{}}{\BconCONT{}{}{, \alpha}{}}
\qquad\infer[\neg \Rightarrow_{\ttt}]{\BconCONT{}{\neg\alpha, }{}{}}{\BconCONT{}{}{}{, \alpha}}
\qquad
\infer[\Rightarrow_{\ttt} \neg]{\BconCONT{}{}{, \neg\alpha}{}}{\BconCONT{\alpha, }{}{}{}}
\qquad
\infer[\Rightarrow_{\ff} \neg]{\BconCONT{}{}{}{, \neg\alpha}}{\BconCONT{}{\alpha, }{}{}}
\\

\newpage
We illustrate the above system with two examples of derivations using the system of $B$-sequents for $\mathbf{E}^B$, namely, derivations for 
$\Bcon{\neg(\alpha\vee\beta)}{}{}{\neg\alpha\wedge\neg\beta}$ and for $\Bcon{}{(\alpha\wedge\beta)\vee\gamma}{(\alpha\vee\gamma)\wedge(\beta\vee\gamma)}{}$:

\vspace{.5cm}

\hfil
\infer[\neg \Rightarrow_{\ttt}]{\Bcon{\neg(\alpha\vee\beta)}{}{}{\neg\alpha\wedge\neg\beta}}{
\infer[\Rightarrow_{\ff} \vee]{\Bcon{}{}{\alpha\vee\beta}{\neg\alpha\wedge\neg\beta}}{
\infer[\Rightarrow_{\ttt} \wedge]{\Bcon{}{}{\alpha,\beta}{\neg\alpha\wedge\neg\beta}}{
\infer[\Rightarrow_{\ttt} \neg]{\Bcon{}{}{\alpha,\beta}{\neg\alpha}}{
\infer[weak]{\Bcon{}{\alpha}{\alpha,\beta}{}}{
\infer[in_f]{\Bcon{}{\alpha}{\alpha}{}}{}}}
\quad
\infer[\Rightarrow_{\ttt} \neg]{\Bcon{}{}{\alpha,\beta}{\neg\beta}}{
\infer[weak]{\Bcon{}{\beta}{\alpha,\beta}{}}{
\infer[in_f]{\Bcon{}{\beta}{\beta}{}}{}}
}}}}
\hfil

\vspace{.5cm}
\hfil
\infer[\Rightarrow_t\wedge]{\Bcon{}{(\alpha\wedge\beta)\vee\gamma}{(\alpha\vee\gamma)\wedge(\beta\vee\gamma)}{}}{
\infer[\Rightarrow_f\vee]{\Bcon{}{(\alpha\wedge\beta)\vee\gamma}{\alpha\vee\gamma}{}}{
\infer[\vee\Rightarrow_f]{\Bcon{}{(\alpha\wedge\beta)\vee\gamma}{\alpha,\gamma}{}}{
\infer[\wedge\Rightarrow_f]{\Bcon{}{\alpha\wedge\beta}{\alpha}{}}{\infer[weak]{\Bcon{}{\alpha, \beta}{\alpha}{}}{\infer[in_f]{\Bcon{}{\alpha}{\alpha}{}}{}}}
\quad
\infer[in_f]{\Bcon{}{\gamma}{\gamma}{}}{}}}
\quad
\infer[\Rightarrow_f\vee]{\Bcon{}{(\alpha\wedge\beta)\vee\gamma}{\beta\vee\gamma}{}}{
\infer[\vee\Rightarrow_f]{\Bcon{}{(\alpha\wedge\beta)\vee\gamma}{\beta,\gamma}{}}{
\infer[\wedge\Rightarrow_f]{\Bcon{}{\alpha\wedge\beta}{\beta}{}}{\infer[weak]{\Bcon{}{\alpha, \beta}{\beta}{}}{\infer[in_f]{\Bcon{}{\beta}{\beta}{}}{}}}
\quad
\infer[in_f]{\Bcon{}{\gamma}{\gamma}{}}{}}}
}
\hfil

We write \infer[]{\Bcon{\Gamma}{\Psi}{\Delta}{\Phi}}{*}, to denote that the $B$-sequent $\Bcon{\Gamma}{\Psi}{\Delta}{\Phi}$ has a derivation using the rules presented above. 


\begin{theor}[Soundness] Every derivable $B$-sequent is valid.
\end{theor}
\begin{proof}[Proof] Induction on the number of $B$-sequent rules applied in a derivation.

\noindent 
Base case: The $B$-sequent is of the form \infer[in_t]{\Bcon{\alpha}{}{}{\alpha}}{}
or
\infer[in_f]{\Bcon{}{\alpha}{\alpha}{}}{}. 

\noindent On the one hand, $\Bcon{\alpha}{}{}{\alpha}$ means that there is no agent $\agn\in SOC$ such that $\acc\agn{:}\alpha$ and $\nacc\agn{:}\alpha$. Thus, $\Bcon{\alpha}{}{}{\alpha}$ is valid.
On the other hand, $\Bcon{}{\alpha}{\alpha}{}$ means that there is not agent $\agn\in SOC$ such that $\nrej\agn{:}\alpha$ and $\rej\agn{:}\alpha$. Thus, $\Bcon{}{\alpha}{\alpha}{}$ is valid.

\noindent
Inductive hypothesis: For any $B$-sequent \infer[]{\Bcon{\Gamma}{\Psi}{\Phi}{\Delta}}{*} whose derivation has up to $k$ application of rules, we have that $\Bcon{\Gamma}{\Psi}{\Phi}{\Delta}$  is valid.

\noindent
Inductive step: Consider a $B$-sequent derived by $k+1$ application of rules.

\noindent Case of [$weak$]. The $B$-sequent is of the form
\infer[weak]{\Bcon{\Gamma',\Gamma}{\Psi',\Psi}{\Phi,\Phi'}{\Delta,\Delta'}}{\infer{\Bcon{\Gamma}{\Psi}{\Phi}{\Delta}}{*}}.
By inductive hypothesis, $\Bcon{\Gamma}{\Psi}{\Phi}{\Delta}$ is valid, and by $B$-monotonicity, $\Bcon{\Gamma',\Gamma}{\Psi',\Psi}{\Phi,\Phi'}{\Delta,\Delta}$ is valid. 

\noindent Case of [$cut_t$]. The $B$-sequent is of the form
\infer[cut_t]{\Bcon{\Gamma}{\Psi}{\Phi}{\Delta}}{
\infer{\BconCONT{}{\alpha, }{}{}}{*'} 
\quad 
\infer{\BconCONT{}{}{, \alpha}{}}{*''}
}.
By inductive hypothesis $\BconCONT{}{\alpha, }{}{}{}$ and $\BconCONT{}{}{, \alpha}{}$ are valid, and by $B$-transitivity(t), we have that $\BconCONT{}{}{}{}$ is valid. The case of [$cut_f$] is similar.

\noindent Case of [$\Rightarrow_{\ttt}\wedge$] The $B$-sequent is of the form 
\infer[\Rightarrow_{\ttt} \wedge]{\BconCONT{}{}{, \alpha\wedge\beta}{}}{
\infer[]{\BconCONT{}{}{, \alpha}{}}{*'} 
\quad 
\infer[]{\BconCONT{}{}{, \beta}{}}{*''}
}. Given the inductive hypothesis, we have that $\BconCONT{}{}{, \alpha}{}$ e $\BconCONT{}{}{, \beta}{}$ are both valid, which means that there is no agent $\agn$ such that $\nacc\agn{:}\alpha$ and $\nacc\agn{:}\beta$. This way, by the recursive clause 2.7, there is no agent $\agn$ such that $\nacc\agn{:}\alpha\wedge\beta$. Thus, $\BconCONT{}{}{, \alpha\wedge\beta}{}$ is valid. 

\noindent Case of [$\Rightarrow_{\ff}\vee$] The $B$-sequent is of the form
\infer[\Rightarrow_{\ff} \vee]{\BconCONT{}{}{}{, \alpha\vee\beta}}{
\infer{\BconCONT{}{}{}{, \alpha,\beta}}{*}
}.  Given the inductive hypothesis, we have that $\BconCONT{}{}{}{, \alpha,\beta}$ is valid, which means that there is no agent $\agn$  such that $\rej\agn{:}\alpha$ or $\rej\agn{:}\beta$. So, by the recursive clause 2.10, there is no agent $\agn$ such that $\rej\agn{:}\alpha\vee\beta$.Thus, $\BconCONT{}{}{}{, \alpha\vee\beta}$ is valid.

\noindent Case of [${\neg}\Rightarrow_{\ttt}$]The $B$-sequent is of the form
\infer[\neg \Rightarrow_{\ttt}]{\BconCONT{}{\neg\alpha, }{}{}}{
\infer[]{\BconCONT{}{}{}{, \alpha}}{*}
}. Given the inductive hypothesis we have that $\BconCONT{}{}{}{, \alpha}$ is valid. This means that there is no agent $\agn$ such that $\rej\agn{:}\alpha$. This way, by the recursive clause 2.1, there is no agent $\agn$ such that $\acc\agn{:}\neg\alpha$. Thus, $\BconCONT{}{\neg\alpha, }{}{}$ is valid.

The other cases are proved in a similar fashion.
\end{proof}

Consider an arbitrary derivation of $B$-sequents \infer{\BconCONT{}{}{}{}}{*}. We call \emph{previous $B$-sequents} all the $B$-sequents that generate $\BconCONT{}{}{}{}$ in the given derivation.

For the completeness result, we will prove first that if a given $B$-sequent is valid, all the previous $B$-sequents on its derivation tree are valid.

\begin{theor}[Inversion theorem] Let $I$ be any inference rule other than weakening. If the conclusion of the application of $I$ is a valid $B$-sequent, then all the $B$-sequents previous to the application of $I$ are valid.  
\end{theor}
\begin{proof}[Proof] Induction on the number of application of rules distinct from $B$-weakening.

\noindent
Base case: The $B$-sequent $\BconCONT{}{}{}{}$ is derived from the rules $in_t$ or $in_f$, thus there are no $B$-sequents previous to $\BconCONT{}{}{}{}$ by the application of any rule other than $B$-weakening.

\noindent
Inductive Hypothesis: All of the previous $B$-sequents of a given valid $B$-sequent$\Bcon{\Gamma}{\Psi}{\Phi}{\Delta}$, whose derivations have up to~$k$ applications of rules other than weakening, are valid.

\noindent
Inductive step: Let $\Bcon{\Gamma}{\Psi}{\Phi}{\Delta}$ be a valid $B$-sequent that may be derived with $k+1$ applications of the rules other than $B$-weakening.

\noindent Case of [$cut_t$]. The valid $B$-sequent has a derivation of the form 
\infer[cut_t]{\BconCONT{}{}{}{}}{
\infer{\BconCONT{}{\alpha, }{}{}}{*'} 
\quad 
\infer{\BconCONT{}{}{, \alpha}{}}{*''}
}.
Since $\BconCONT{}{}{}{}$ is valid, there is no agent $\agn$ such that  $\acc s{:}\Gamma$ and $\nacc s{:}\Delta$ and $\rej s{:}\Phi$ and $\nrej s{:}\Psi$. By $B$-monotonicity, we have that there is no agent $\agn$ such that $\acc s{:}\alpha,\Gamma$ and $\nacc s{:}\Delta$ and $\rej s{:}\Phi$ and $\nrej s{:}\Psi$. Similarly, there is no agent $\agn$ such that $\acc s{:}\Gamma$ and $\nacc s{:}\Delta,\alpha$ and $\rej s{:}\Phi$ and $\nrej s{:}\Psi$. Thus, $\Bcon{\alpha,\Gamma}{\Psi}{\Phi}{\Delta}$ and $\Bcon{\Gamma}{\Psi}{\Phi}{\Delta,\alpha}$ are both valid, and, by inductive hypothesis, all of their previous $B$-sequents are valid.
The case of [$cut_f$] is similar.

\noindent Case of [$\Rightarrow_{\ttt}\wedge$] The valid $B$-sequent has a derivation of the form 
\infer[\Rightarrow_{\ttt} \wedge]{\BconCONT{}{}{, \alpha\wedge\beta}{}}{
\infer[]{\BconCONT{}{}{, \alpha}{}}{*'} 
\quad 
\infer[]{\BconCONT{}{}{, \beta}{}}{*''}
}. 
Since $\BconCONT{}{}{, \alpha\wedge\beta}{}$ is valid, there is no agent $\agn$ such that  $\acc s{:}\Gamma$ and $\nacc s{:}\Delta, \alpha\wedge\beta$ and $\rej s{:}\Phi$ and $\nrej s{:}\Psi$. By the recursive clause 2.7, we have that there is no agent $\agn$ such that $\nacc s{:}\alpha\wedge\beta$, since there is no agent $\agn$ such that $\nacc\agn{:}\alpha$ and $\nacc\agn{:}\beta$. This way, there is no agent $\agn$ such that $\acc s{:}\Gamma$ and $\nacc s{:}\Delta, \alpha$ and $\rej s{:}\Phi$ and $\nrej s{:}\Psi$, and also there no agent $\agn$ such that $\acc s{:}\Gamma$ and $\nacc s{:}\Delta, \beta$ and $\rej s{:}\Phi$ and $\nrej s{:}\Psi$. Thus, $\BconCONT{}{}{, \alpha}{}$ and $\BconCONT{}{}{, \beta}{}$ are both valid, and, by inductive hypothesis, all of their previous $B$-sequents are valid.

\noindent Case of [$\Rightarrow_{\ff}\vee$] The valid $B$-sequent has a derivation of the form 
\infer[\Rightarrow_{\ff} \vee]{\BconCONT{}{}{}{, \alpha\vee\beta}}{
\infer{\BconCONT{}{}{}{, \alpha,\beta}}{*}
}. 
Since $\BconCONT{}{}{}{, \alpha\vee\beta}$ is valid, there is no agent $\agn$ such that $\acc s{:}\Gamma$ and $\nacc s{:}\Delta$ and $\rej s{:}\Phi, \alpha\vee\beta$ and $\nrej s{:}\Psi$. By the recursive clause 2.10, there is no agent $\agn$ such that $\rej\agn{:}\alpha\vee\beta$ if $\rej\agn{:}\alpha$ or $\rej\agn{:}\beta$, thus $\BconCONT{}{}{}{, \alpha,\beta}$ is valid, and, by inductive hypothesis, all of its previous $B$-sequents are valid.

\noindent Case of [${\neg}\Rightarrow_{\ttt}$] The valid $B$-sequent has a derivation of the form 
\infer[\neg \Rightarrow_{\ttt}]{\BconCONT{}{\neg\alpha,}{}{}}{
\infer[]{\BconCONT{}{}{}{, \alpha}}{*}
}.
Since $\BconCONT{}{\neg\alpha,}{}{}$ is valid, there is no agent $\agn$ such that $\acc s{:}\neg\alpha, \Gamma$ and $\nacc s{:}\Delta$ and $\rej s{:}\Phi$ and $\nrej s{:}\Psi$. By the recursive clause 2.1, there is no $\agn$ such that $\acc\agn{:}\neg\alpha$ if $\rej\agn{:}\alpha$. Thus, $\BconCONT{}{}{}{, \alpha}$ is valid.

The other cases have similar proofs.
\end{proof}

\begin{theor}[Completeness]\label{teo:compl} If $\BconCONT{}{}{}{}$ is a valid $B$-sequent in $\mathbf{E}^B$, then there is a derivation of $\BconCONT{}{}{}{}$.
\end{theor}

\begin{proof}[Proof] Induction on the number of connectives of the $B$-sequent $\BconCONT{}{}{}{}$.

\noindent
Base case: $\BconCONT{}{}{}{}$ has no logical connectives. In this case, all of the formulas are propositional variables. Given that $\BconCONT{}{}{}{}$ is valid, there must be some propositional variable $\alpha$ such that either $\alpha\in\Gamma\cap\Delta$ or $\alpha\in\Phi\cup\Psi$. Therefore $\BconCONT{}{}{}{}$ may be proved either by \infer[in_t]{\Bcon{\alpha}{}{}{\alpha}}{} or by \infer[in_f]{\Bcon{}{\alpha}{\alpha}{}}{} and, eventually, by applications of the $B$-weakening rule.

\noindent
Inductive hypothesis: Let $\BconCONT{}{}{}{}$ be a valid $B$-sequent which has up to $m\in\mathbb{N}$ occurrences of logical connectives. Then, there is a derivation of $\BconCONT{}{}{}{}$ free of applications of $B$-cut.

The inductive step, $m+1$, is obtained by cases according to the most external connective of the $B$-sequent formulas.

\noindent
[Case $\neg\alpha\in\Gamma$] Let $\Gamma'$ be the set of formulas obtained from $\Gamma$ by removing all of the occurrences of $\neg\alpha$. We may then infer $\BconCONT{}{}{}{}$ by:
\infer={\BconCONT{}{}{}{}}{\infer[\neg_f\Rightarrow_{\tt}]{\Bcon{\neg\alpha, \Gamma'}{\Psi}{\Delta}{\Phi}}{\infer{\Bcon{\Gamma'}{\Psi}{\Phi, \alpha}{\Delta}}{*}}}, where the double line indicates consecutive applications of $B$-weakening.
By the inversion theorem, $\Bcon{\Gamma'}{\Psi}{\Phi, \alpha}{\Delta}$ is valid, and, since it has at most~$m$ logical connectives, the inductive hypothesis implies the existence of a derivation of this $B$-sequent without the use of $B$-cuts. 
By applying the rule [${\neg}\Rightarrow_f$], we have $\Bcon{\neg\alpha, \Gamma'}{\Psi}{\Delta}{\Phi}$, whose derivation is free of $B$-cuts. 
The proofs for the cases where a formula of the form $\neg\alpha$ occurs in $\Delta, \Phi$ and $\Psi$ are similar to the present first case.

\noindent
[Case $\alpha\wedge\beta\in\Gamma$] Let $\Gamma'$ be the set of formulas obtained from $\Gamma$ by removing all of the occurrences of $\alpha\wedge\beta$. We may then infer $\BconCONT{}{}{}{}$ by:
\infer={\BconCONT{}{}{}{}}{\infer[\wedge\Rightarrow_t]{\Bcon{\alpha\wedge\beta, \Gamma'}{\Psi}{\Phi}{\Delta}}{\infer{\Bcon{\alpha,\beta, \Gamma'}{\Psi}{\Phi}{\Delta}}{*}}}. By the inversion theorem, $\Bcon{\alpha,\beta, \Gamma'}{\Psi}{\Phi}{\Delta}$ is valid, and since it has at most $m$ logical connectives, the inductive hypothesis implies the existence of a derivation without the use of $B$-cuts. 
By applying the rule [$\wedge\Rightarrow_t$], we have $\BconCONT{}{\alpha\wedge\beta, }{}{}$. 
The proofs for the cases where $\alpha\wedge\beta\in\Psi$, $\alpha\vee\beta\in\Phi$ and $\alpha\vee\beta\in\Delta$ are similar to the present case.

\noindent
[Case $\alpha\vee\beta\in\Gamma$] Let $\Gamma'$ be the set of formulas obtained from $\Gamma$ by removing all of the occurrences of $\alpha\vee\beta$, we can infer $\BconCONT{}{}{}{}$ by:
\infer={\BconCONT{}{}{}{}}{\infer[\vee\Rightarrow_t]{\Bcon{\alpha\vee\beta, \Gamma'}{\Psi}{\Phi}{\Delta}}{\infer{\Bcon{\alpha, \Gamma'}{\Psi}{\Phi}{\Delta}}{*'}\quad\infer{\Bcon{\beta, \Gamma'}{\Psi}{\Phi}{\Delta}}{*''}}}. By the inversion theorem, $\Bcon{\alpha,\Gamma'}{\Psi}{\Phi}{\Delta}$ and $\Bcon{\beta,\Gamma'}{\Psi}{\Phi}{\Delta}$ are valid and as the sum of their connectives is less than $m$, so by the inductive hypothesis we know that there are derivations of the latter without the use of $B$-cuts. 
By applying the rule [$\vee\Rightarrow_t$], we have $\BconCONT{}{\alpha\vee\beta, }{}{}$, whose derivation is free of the rules of $B$-cuts.
The proofs for the cases where $\alpha\vee\beta\in\Psi$, $\alpha\wedge\beta\in\Phi$ and $\alpha\wedge\beta\in\Delta$ are similar to the present case.
\end{proof}

The logic $\mathbf{E}^B$, inspired by {First Degree Entailment}, is associated to a {$B$-entailment}, a consequence relation which has four positions allowing for reasoning with incomplete or inconsistent information. By adopting {$B$-entailment} as a semantic consequence relation in the logic $\mathbf{E}^B$, we allow for a logic which expresses different forms of reasoning in terms of acceptance (or not) and rejection (or not) of statements, without the truth being in any way privileged over the falsity 
in its inferences.
This privilege is found in the original Dunn-Belnap logic~$\mathbf{E}$, whose Tarskian consequence relation is defined solely in terms of the preservation of truth-values containing the truth.
We thus believe that this new definition of consequence relation is more adequate for a formalism that intends to deal with inconsistent and partial informational content.




\begin{spacing}{0.8}

\end{spacing}
\end{document}